\newcommand{\CLASSINPUTtoptextmargin}{0.78 in}
\newcommand{\CLASSINPUTbottomtextmargin}{1.06 in}

\documentclass[conference]{IEEEtran}
\makeatletter
\def\ps@headings{%
\def\@oddhead{\mbox{}\scriptsize\rightmark \hfil \thepage}%
\def\@evenhead{\scriptsize\thepage \hfil \leftmark\mbox{}}%
\def\@oddfoot{}%
\def\@evenfoot{}}
\makeatother
\usepackage{graphicx}
\usepackage{subfigure}
\usepackage[english]{babel}
\usepackage{url}            
\usepackage{booktabs}       
\usepackage{amsfonts}       
\usepackage{nicefrac}       
\usepackage{microtype}      
\usepackage{lipsum}
\usepackage{amsthm}
\usepackage{amssymb}
\usepackage{amsmath}
\usepackage{mathptmx}
\usepackage{hyperref}       
\usepackage{xcolor}
\usepackage{comment}
\usepackage{soul}

\usepackage{algorithm}
\usepackage[noend]{algpseudocode}
\makeatletter
\def\BState{\State\hskip-\ALG@thistlm}
\makeatother

\usepackage[english]{babel}

\DeclareMathOperator*{\argmin}{arg\,min}

\newcommand{\dd}{\mathrm{d}}

\newcommand{\ie}{\emph{i.e.}}

\newtheorem{theorem}{Theorem}
\newtheorem{lemma}{Lemma}

\newtheorem{example}{Example}
\newtheorem{definition}{Definition}
\newtheorem{remark}{Remark}

\usepackage{comment}

\begin{document}
%
\title{On the Role of Risk Perceptions in Cyber Insurance Contracts}
%
%
%

\author{Shutian~Liu
        and~Quanyan~Zhu
\thanks{The authors are with the Department of Electrical and Computer Engineering, Tandon School of Engineering,
        New York University, Brooklyn, NY, 11201 USA (e-mails: sl6803@nyu.edu; qz494@nyu.edu).}       
}

\pagestyle{empty}

\maketitle
\thispagestyle{empty}


\begin{abstract}
Risk perceptions are essential in cyber insurance contracts.
With the recent surge of information, human risk perceptions are exposed to the influences from both beneficial knowledge and fake news.
In this paper, we study the role of the risk perceptions of the insurer and the user in cyber insurance contracts.
We formulate the cyber insurance problem into a principal-agent problem where the insurer designs the contract containing a premium payment and a coverage plan.
The risk perceptions of the insurer and the user are captured by coherent risk measures.
Our framework extends the cyber insurance problem containing a risk-neutral insurer and a possibly risk-averse user, which is often considered in the literature.
The explicit characterizations of both the insurer's and the user's risk perceptions allow us to show that cyber insurance has the potential to incentivize the user to invest more on system protection.
This possibility to increase cyber security relies on the facts that the insurer is more risk-averse than the user (in a minimization setting) and that the  insurer's risk perception is more sensitive to the changes in the user's actions than the user himself.
We investigate the properties of feasible contracts in a case study on the insurance of a computer system against ransomware.

\end{abstract}


%

\section{Introduction}
\label{sec:intro}
%
%
%
%

Human risk perception plays an important role in cyber insurance contracts.
On the one hand, individuals who are risk-averse tend to overreact to severe potential losses that are not likely to happen. 
On the other hand, they are eager to seek additional resources to defend against cyber losses. 
The risk-sharing property of cyber insurance contracts that mitigates the cyber losses of users depends on the user's risk-aversion. The cyber insurance market does not exist when the users are risk-neutral \cite{marotta2017cyber}.
In an era of information explosion, people may either intend to adjust their risk attitudes according to expert advice or be manipulated by fake news to reform they risk preferences without awareness \cite{liu2022mitigating}. 
The instability of risk perception can have essential impacts on the insurance market and the resiliency of cyber systems. Therefore, there is a need to study how users with different risk perceptions behave when they face potential cyber losses and how the optimal insurance plan should change according to the risk attitudes.

Linear contracts involving a prepaid premium and a coverage plan are among the most considered contract models in the cyber insurance market. The premium is a money transfer from the user to the insurer for entering the contract and the coverage plan describes the proportion of cyber losses covered by the insurer.  
In a linear contract model, the insurer is often assumed to be risk-neutral and evaluate her loss using expectations; the user is set to exhibit risk-aversion.
The risk-aversion of an individual can be modeled by a nonlinear utility function \cite{tversky1992advances,stole2001lectures}.
The risk-aversion is captured by the fact that the utility function is assumed to be increasing and concave in return.

There is a recent trend on the study risk quantification adopting coherent risk measures (CRMs). The axiomatic definition of CRMs maintains the generality of the risk modeling and provides rich insights towards applications. The dual representation of a CRM shows its robustness to probabilistic uncertainty. Reformulation techniques \cite{rockafellar2000optimization,acerbi2002spectral,shapiro2013kusuoka} have also enabled convenient and tractable computation methods of risks modeled by various CRMs.

\begin{figure}[ht]
\centering
\vspace{-3.2mm}\includegraphics[width=0.46\textwidth]{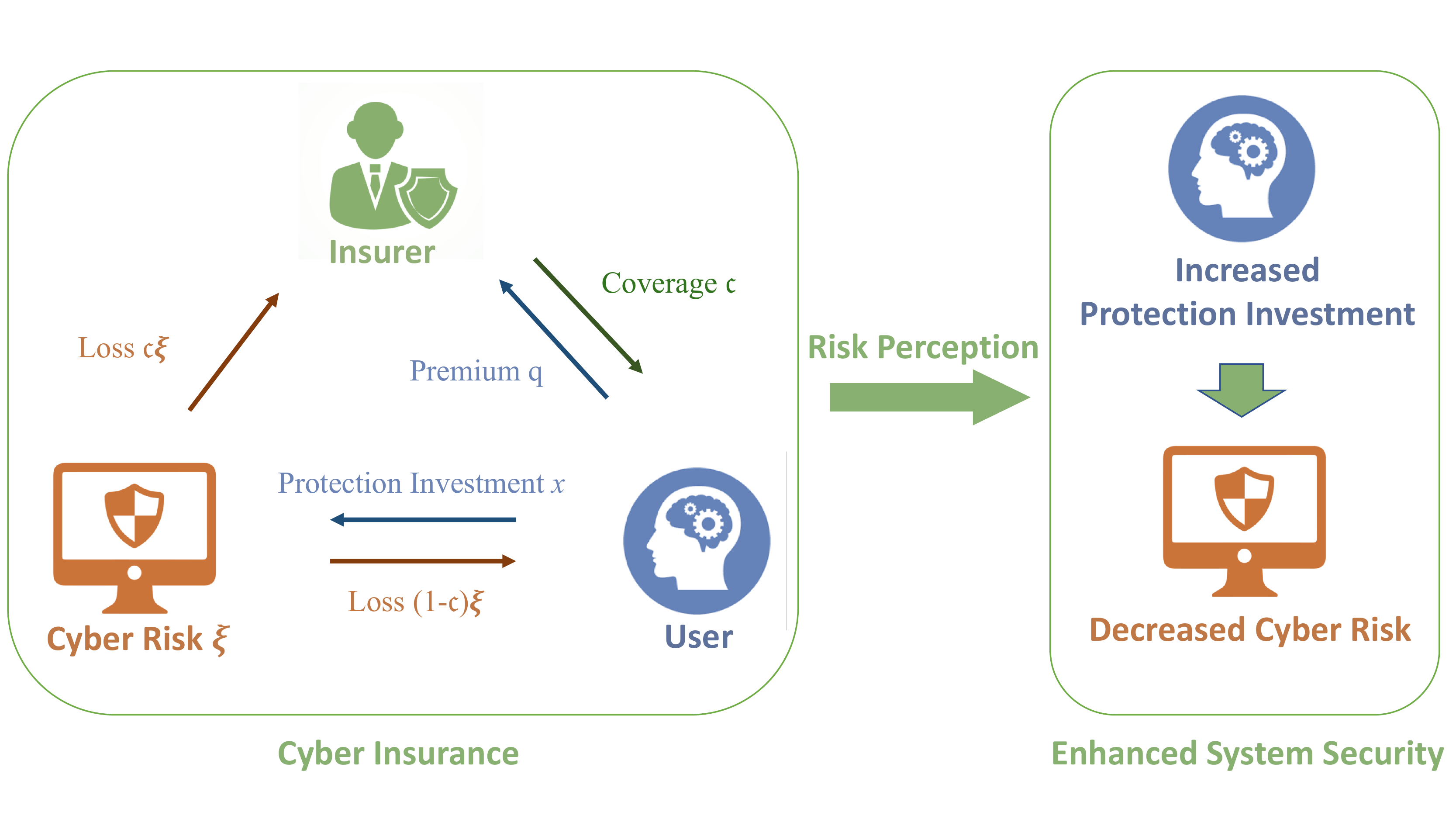}

\vspace{-2mm}\caption[Optional caption for list of figures]{Cyber insurance has the potential to enhance system security by incentivizing the user to invest more on system protection, if the risk perception of the insurer exhibits more risk-aversion and is more sensitive towards the distributional shifts of the cyber losses.  } 
\label{fig:framework}
\end{figure}

We study the role of human risk perceptions in cyber insurance using a holistic framework which incorporates the modern risk modeling approach and a linear principal-agent (P-A) model.
In particular, we use CRMs to describe the risk-aversion of the principal and the agent. 
The reason is two-fold.
First, CRMs allow us to investigate the probabilistic distortion to the random cyber losses caused by human risk-aversion.
Second, cyber risks are challenging to quantify due to the difficulty in transforming the cyber losses to monetary losses.
Therefore, the probabilistic robustness that CRMs possess leads to reliable and safe estimations of cyber risks. 
Our framework builds on the hidden-action linear contract problem to capture the information asymmetry between the insurer and the user. Specifically, the principal minimizes her loss function by designing the contract containing a premium and a coverage rate subject to the individual rationality (IR) constraint which guarantees beneficial participation and the incentive compatibility (IC) constraint which corresponds to the rationality of the agent. Due to the IC constraint, the contract problem appears in the form of a bilevel program. Though a full-information counterpart to this problem where the IC constraint is absent produces a lower optimal loss, the hidden-information does not suffer from the moral hazard issue \cite{holmstrom1979moral}.
In this work, we focus on the influence of risk measures on the insurance. Hence, we do not consider additional nonlinearities on top of the random cyber loss modeled by a random variable endowed with a parametric distribution. The validity of linear contracts follows from the monotonicity property of the solution to general contract problems \cite{holmstrom1979moral}.

Hidden-action contract problems are challenging because of the bilevel nature. However, leveraging the linearity of the contract and the first-order approach, we can simplify the problem and derive its optimality conditions.
The conditions allow us to characterize the coverage and the premium in terms of the derivative of the risk of the user with respect to his action. By choosing proper risk measures, practitioners can obtain optimal contracts which satisfy desired properties.

One of the most essential features of the principal-agent models lies in that the distribution of the random losses is parametric in the agent’s action. In our framework, how the risks perceived by the insurer and the user change according to the user' actions captures how sensitive the insurer and the user are towards the parameterization.
We show that when, compared to the user's risk perception, the insurer's risk perception exhibits more aversion to random cyber losses and is more sensitive to the parameterization,
cyber insurance can enhance system security by incentivizing the protection investment of the agent. 
These requirements suggest the following characteristics of the insurer.
First, the insurer, who bears the responsibility in evaluating the system risks, should be able to estimate the cyber losses more cautiously than the user.
Second, aiming to design an incentive contract, the insurer should possess a higher level of awareness of how the actions from the user influence the system risks stochastically than the user himself.
Our result enriches the literature by introducing the possibility that cyber insurance can incentivize the user's system protection investment and hence enhance the overall system security. This possibility is not observed in traditional cyber insurance problems where the risk perceptions are captured by nonlinear utility functions \cite{marotta2017cyber,khalili2018designing,bohme2010modeling}.

The paper is organized as follows. In Section \ref{sec:problem}, we first introduce the risk preference modeling, then we incorporate it into the cyber insurance contract design problem.
Section \ref{sec:analysis} contains the analysis of the game. We discuss the roles of risk perceptions in shaping the optimal contract and the relation between risk sensitivity and system security.
We use a case study to further investigate the insurance contracts in Section \ref{sec:case}.
Finally, Section \ref{sec:conclusion} concludes the paper.

\section{Problem Formulation}
\label{sec:problem}
In this section, we first introduce the definition of CRMs and their analytical properties. Then, we introduce the cyber insurance contract design problem with the risk preferences of the insurer and the user described by CRMs.

\subsection{Risk Preference Modeling}
\label{sec:problem:risk}
Consider the probability space $(\Xi, \mathcal{F})$ of cyber loss samples $\xi\in\Xi\subset \mathbb{R_+}$ endowed with the reference probability measure $P$.
Let $\mathcal{Z} :=\mathcal{L}_p(\Xi, \mathcal{F}, P)$ denotes the space of random losses $Z:\Xi\rightarrow \mathbb{R}$ with finite $p$-th order moment. 
The parameter $p$ lives in $[1,+\infty)$.
A risk measure $\rho$ is a function $\mathcal{Z}\rightarrow \mathbb{R}$ that assigns a deterministic value to a random loss. 
Classic approaches to risk modeling includes using the expected loss, the standard deviation of the loss, the value-at-risk, and etc. 
These risk metrics can come in handy in many real situations due to their simplicity and straightforwardness of interpretation. 
However, the classic risk metrics are lacking in the following two ways.
First, one risk metric cannot fully characterize the behavior of a random loss. A simple example would be that using the expectation to characterize the risk of a Gaussian random loss has $50\%$ chance to fail when the randomness is realized.
Second, human risk perceptions are different across individuals. According to \cite{tversky1992advances}, humans tend to distinguish between losses and gains and are likely to perceive the true probability of random events with biases.
Hence, risk metrics should characterize human behaviors beyond merely risk-neutrality.

In this paper, we will use CRMs to characterize the risk sensitivities that the insurer and the user exhibit. 
\begin{definition}[Coherent Risk Measures \cite{artzner1999coherent}]
\label{def:CRM}
A function $\rho:\mathcal{Z}\rightarrow\mathbb{R}$ is called a Coherent Risk Measure if for $Z, Z' \in \mathcal{Z}$ it satisfies
\\
(A1) Monotonicity: $\rho(Z)\geq \rho(Z')$ if $Z(\xi)\geq Z'(\xi)$ for almost everywhere $\xi\in \Xi$.\\
(A2) Convexity: 
    $\rho(tZ+(1-t)Z') \leq t\rho(Z)+(1-t)\rho(Z')$
for $t\in[0,1]$.\\
(A3) Translation equivariance: $\rho(Z+a)=\rho(Z)+a$ if $a\in\mathbb{R}$.\\
(A4) Positive homogeneity: $\rho(tZ)=t\rho(Z)$ if $t\geq 0$.
\end{definition}
One definition of risk-aversion \cite{uryasev2010var} is referred to the fact that the perceived risk is not smaller than the expectation of the random loss, \ie, $\rho[\cdot]\geq \mathbb{E[\cdot]}$.
A convex risk measure captures the risk-aversion of decision-makers in this sense.

A CRM captures the decision-maker's robustness consideration to probabilistic uncertainty due to the following dual representation \cite{artzner1999coherent,shapiro2021lectures}:
\begin{equation}
    \rho[Z(\xi)]=\sup_{\zeta \in \mathfrak{A}}
    \int_{\Xi}Z(\xi)\zeta(\xi)\dd P(\xi),
    \label{eq:dual representation of individual risk}
\end{equation}
where $\mathfrak{A}\subset \mathcal{Z}^*$ denotes the dual set associated with the risk measure $\rho_\theta$ and contains probability density functions with respect to the probability measure $P$. The set $\mathcal{Z}^*$ denotes the dual space of $\mathcal{Z}$ defined by $\mathcal{Z}^*:\mathcal{L}_q(\Xi, \mathcal{F}, P)$ with $\frac{1}{p}+\frac{1}{q}=1$.
The optimization problem (\ref{eq:dual representation of individual risk}) admits an optimal solution since the set $\mathfrak{A}$ is convex and compact when $p\in[1,+\infty)$ \cite{shapiro2021lectures}.

The following is an important property of a risk measure.
\begin{definition}
\label{def:law invariance}
(Law-invariance.)
A risk measure $\rho:\mathcal{Z}\rightarrow \mathbb{R}$ is law-invariant with respect to the reference probability measure $P$, if $\forall Z_1,Z_2\in \mathcal{Z}$ such that $P(Z_1 \leq t)=P(Z_2 \leq t)$ for all $t\in\mathbb{R}$, then $\rho(Z_1)=\rho(Z_2)$.
\end{definition}
A law-invariant risk measure maps random variables that admit the same cumulative distribution function to the same risk.
In other words, $\rho(\cdot)$ is law-invariant, the function $\rho(Z)$ only depends on the cumulative distribution function of $Z$.

The following are some examples of risk measures.
\begin{example}
(Average value-at-risk.) Let $F_Z(t)$ denotes the cumulative distribution function of a random variable $Z$. The average value-at-risk (AV@R) at $a\%$-level is defined as:
\begin{equation}
    \text{AV@R}_a(Z):=(1-a)^{-1}\int_{a}^{1}\text{V@R}_t(Z)\dd t, a\in[0,1),
    \label{eq:avar}
\end{equation}
where $\text{V@R}_b(Z):= \inf \{t: F_Z(t)\geq b\}$.
AV@R is a law-invariant CRM \cite{ruszczynski2006optimization}.
\end{example}

\begin{example}
(Absolute semideviation.) The absolute semideviation of a random variable is defined as $\rho_\theta[Z]=\mathbb{E}[Z]+\theta \mathbb{E}\{[z-\mathbb{E}(Z)]_+\}$ with $\theta\geq0$. 
The absolute semideviation is a law-invariant CRM if $\theta\in [0,1]$ \cite{shapiro2013kusuoka}.
\end{example}

\subsection{Cyber Insurance Contract Design Problem}
\label{sec:problem:insurance}
Our contract design problem involves one insurer and one user.
Let $x\in X$ denotes the protection investment of the user, where $X \subset \mathbb{R}_+$ is convex and compact.
The cost of investment is $m>0$.
The protection investment affects the distribution of the cyber loss $\xi$.
In particular, we assume that when the user takes action $x$, $\xi$ admits a parameterized cumulative distribution function $P(\xi,x)$ which is twice differentiable with respect to $x$ and absolutely continuous with respect to $P(\xi)$.
The probability density function of $P(\xi,x)$ with respect to $P(\xi)$ is denoted $p(\xi, x)$.
Furthermore, we make the following two assumptions.
Firstly, we assume the protection investment benefits the system in the sense that the investment provides a first-order stochastic dominance shift on $\Xi$, \ie, for all $x_1, x_2\in X$ such that $x_1\leq x_2$, the distribution functions satisfy $P(\xi,x_1) \leq P(\xi, x_2)$ for all $\xi\in \Xi$.
Intuitively, this stochastic dominance assumption means that when the protection investment is higher, the likelihood that severer cyber losses occur is lower.
Secondly, we assume that the density $p(\xi,x)$ is convex in $x$ for all $\xi\in\Xi$, \ie, $\frac{\partial^2}{\partial x^2}p(\xi,x)\geq 0$.
Since we are considering minimization problems, this assumption will lead to convex optimization problems as we will discuss later.
Furthermore, this convexity assumption indicates
that the user needs to successively increase his protection investment to obtain remarkable mitigation of cyber losses.
Hence, the user has to balance his protection investment between the increasing cost of system updates and the potential decreased likelihood of severe cyber losses.

When the user does not purchase the cyber insurance, he receives the full random cyber loss $Z=\xi$ distributed according to the parameterized distribution $P(\xi,x)$.
Note that one can consider a more general loss function, such as $Z=\phi(\xi)$ where $\phi(\cdot)$ is a nonlinear continuous function.
However, since our focus is on the role of risk measures, we choose $\phi$ as the identity function. 
Our results can be generalized with slightly more complicated notations.

The insurer designs the coverage plan $c\in [0,1]$ and the premium payment $q>0$ in the cyber insurance contract. 
When the user chooses to purchase the insurance, the user only receives $(1-c)\xi$ of the random cyber loss at a cost of paying the premium $q$ to the insurer.
The rest amount of the random cyber loss $c\xi$ is beard by the insurer.

We assume that the risk perceptions of the insurer and the user are captured by risk measures $\rho^i$ and $\rho^u$, respectively.
When the user purchases the insurance, his loss function is: $U(x):=\rho^u[(1-c)\xi+mx+q]$.
The loss function of the insurer is:
$J(c,p,x):=\rho^i[c\xi-q]$.
The cyber insurance contract design problem can formulated as the following Stackelberg game:
\begin{equation}
    \begin{aligned}
         \min_{c\in [0,1],q>0,x\in X} \  \ &\rho^i[c\xi-q] \\
         \text{s.t.} \  \ &\rho^u[(1-c)\xi+mx+q]\leq \Bar{U},\  \ \text{(IR)}, \\
         & x\in \argmin_{x'\in X} \rho^u[(1-c)\xi+mx'+q], \  \ \text{(IC)}.
         \label{eq:hidden action contract problem}
    \end{aligned}
\end{equation}
where $\Bar{U}$ is the optimal solution of the optimization problem:
\begin{equation}
    \min_{x\in X} \rho^u[\xi +mx],
    \label{eq:action without insurance}
\end{equation}
which the user faces when he does not purchase the insurance. 
The IR constraint in (\ref{eq:hidden action contract problem}) guarantees that purchasing the insurance is  beneficial to the user.
The IC constraint in (\ref{eq:hidden action contract problem}) takes into account the user's rationality.
Note that the contract problem (\ref{eq:hidden action contract problem}) is called the hidden-action P-A problem.
The term "hidden-action" refers to the fact that the insurer is unable to observe the true action taken by the user after he purchases the insurance.
Therefore, the IC constraint appears as the insurer's ex ante anticipation of the user's action.
The hidden-action P-A problem is free of the moral hazard issue \cite{stole2001lectures}.

The optimal coverage rate and the optimal premium obtained as the solution to problem (\ref{eq:hidden action contract problem}) are denoted $c^*$ and $q^*$.

\section{The Role of Risk Preferences}
\label{sec:analysis}
In this section, we investigate the role of the risk preferences of the insurer and the user in contract design.
We assume that the risk measures $\rho^u$ and $\rho^i$ are coherent and law-invariant.

\subsection{Problem Simplification}
\label{sec:analysis:simplification}
Since $\rho^u$ and $\rho^i$ satisfy axioms (A3) and (A4), (\ref{eq:hidden action contract problem}) becomes:
\begin{equation}
    \begin{aligned}
         \min_{c\in[0,1],q>0,x\in X} \  \ &c\rho^i[\xi]-q \\
         \text{s.t.} \  \ &(1-c)\rho^u[\xi]+mx+q\leq \Bar{U},\  \ \text{(IR)}, \\
         & x\in \argmin_{x'\in X} (1-c)\rho^u[\xi]+mx+q, \  \ \text{(IC)}.
         \label{eq:hidden action contract problem 1}
    \end{aligned}
\end{equation}
The risks $\rho^u[\xi]$ and $\rho^i[\xi]$ are functions of the user's investment $x$, for the random cyber loss $\xi$ follows the distribution $P(\xi,x)$.

The IR constraint in (\ref{eq:hidden action contract problem 1}) is binding. 
The reason lies in that the insurer can decrease her objective value by increasing the premium $q$ without affecting the IC constraint.
Hence, we obtain the following equality:
\begin{equation}
    (1-c)\rho^u[\xi]+mx+q=\Bar{U}.
    \label{eq: IR is binding}
\end{equation}
Using (\ref{eq: IR is binding}), problem (\ref{eq:hidden action contract problem 1}) can then be further simplified as follows:
\begin{equation}
    \begin{aligned}
         \min_{c\in[0,1],x\in X} \  \ &c\rho^i[\xi]+(1-c)\rho^u[\xi]+mx \\
         \text{s.t.} \  \
         & x\in \argmin_{x'\in X} (1-c)\rho^u[\xi]+mx, \  \ \text{(IC)}.
         \label{eq:hidden action contract problem 2}
    \end{aligned}
\end{equation}
Assuming that the solution to IC is in the interior of $X$, we obtain the first-order optimality condition of IC in (\ref{eq:hidden action contract problem 2}):
\begin{equation}
    0=(1-c)\frac{\partial}{\partial x} \rho^u[\xi]+m,
    \label{eq:foc of IC}
\end{equation}
where the term $\frac{\partial}{\partial x}
\rho^u[\xi]$ is the sensitivity of the risk of the user with respect to his protection investment $x$.
Note that the (sub)differentiability of a risk measure at a point is guaranteed when it is finite and continuous at that point and conditions (A1) and (A2) are satisfied \cite{shapiro2021lectures}.
Since CRMs admit the optimization formulation (\ref{eq:dual representation of individual risk}), Danskin's theorem indicates that the subdifferential $\partial \rho^u[\xi]$ takes the following form:
\begin{equation}
    \partial \rho^u[\xi]=\text{cov}\left\{\cup_{\zeta  \in \mathfrak{A}^*_u} \int_{\Xi}\partial(\xi p(\xi,x))\zeta(\xi)\dd P(\xi)\right\},
    \label{eq:subdifferential of rho}
\end{equation}
where $\text{cov}\{\cdot\}$ denotes the convex hull and $\mathfrak{A}^*_u$ denotes the set of optimizers of the dual representation of $\rho^u[\xi]$.
We refer the readers to \cite{shapiro2021lectures} for more details regarding the differentiability of risk measures.
The formula (\ref{eq:subdifferential of rho}) will be useful in Section \ref{sec:analysis:security} when we study the relation between risk perceptions and system security.
Combining (\ref{eq:foc of IC}) and (\ref{eq:hidden action contract problem 2}), we obtain the following reformulation of the contract design problem:
\begin{equation}
    \min_{c\in [0,1], x\in X} c\rho^i[\xi]+(1-c)\rho^u[\xi]+x\left(  (c-1)\frac{\partial}{\partial x}\rho^u[\xi] \right) -\Bar{U}.
    \label{eq:hidden action contract problem 3}
\end{equation}

\subsection{Properties of Feasible Contracts}
\label{sec:analysis:optimal contract}
From the discussions in Section \ref{sec:analysis:simplification}, we observe that a feasible contract satisfies the IR and the IC constraints.
Combining (\ref{eq:foc of IC}) and (\ref{eq: IR is binding}), we find that a feasible coverage plan $c$ and a feasible premium $q$ satisfies the following equations:
\begin{equation}
    c=1+\frac{m}{\frac{\partial}{\partial x}\rho^u[\xi]\big|_{x}},
    \label{eq:c}
\end{equation}
and 
\begin{equation}
    q=\Bar{U}-mx-c\rho^u[\xi] \big|_{x},
    \label{eq:q}
\end{equation}
where $x\in X$.
From (\ref{eq:c}), we observe that the feasible coverage rate is a decreasing function of $\frac{\partial}{\partial x}\rho^u[\xi]$.
The user's action completely determines the coverage.
The feasible premium in (\ref{eq:q}) also directly depends on the the user's risk.

We invoke the following result from \cite{ruszczynski2006optimization}.
\begin{lemma}
\label{lemma:consistency}
If a risk measure $\rho:\mathcal{Z}\rightarrow \mathbb{R}$ is law-invariant and satisfies condition (A1), then it is consistent with first-order stochastic dominance, \ie, $\rho[Z_1]\geq \rho[Z_2]$ if and only if $Z_1$ stochastically dominates $Z_2$ in the first-order.
\end{lemma}
Since we have assumed that when $x_2\geq x_1$, $P(\xi, x_1)$ stochastically dominates $P(\xi, x_2)$, the consistency result of Lemma \ref{lemma:consistency} indicates that $\frac{\partial}{\partial x}\rho^u[\xi]\leq 0$ and $\frac{\partial}{\partial x}\rho^i[\xi]\leq 0$.

Applying \ref{lemma:consistency} on (\ref{eq:c}), we arrive at a direct consequence that feasible coverage rates is increasing in the user's actions.
When utilized in the inverse direction, this monotonicity property suggests that the insurer can potentially incentivize the user to invest more on protection by increasing the coverage rate of the cyber insurance.
The monotonicity of feasible premiums with respect to the user's action is challenging to conclude solely from (\ref{eq:q}).
The reason lies in that the second and third terms on the right-hand side of (\ref{eq:q}) evolve to the opposite directions when the user's action changes. 

When $\rho^u[\xi]$ and $\frac{\partial}{\partial x}\rho^u[\xi]$ are evaluated at points in $X$ which optimize the objective function of (\ref{eq:hidden action contract problem 2}), (\ref{eq:c}) and (\ref{eq:q}) arrive at the optimal coverage rate and premium.
Hence, the values of the optimal coverage and the optimal premium can be calibrated conveniently when the set $\mathfrak{A}^*_u$ is a singleton.
If one is aware of how user's risk perception changes under exogenous influences, such as marketing, nudging, and information campaigns, then one can achieve better performances in the optimal contract using the knowledge.
Since the effect of information on human risk perception is beyond the scope of the current paper, we leave this topic as a future work.

\subsection{The Influence of Risk Perceptions on System Security}
\label{sec:analysis:security}

Cyber insurance, as a tool to increase the resiliency of cyber systems, is the last resort to enhance the system security.
The overall system security largely depends on the how much the user invests on protection.
In this part, we investigate the risk preferences of the insurer and the user to provide conditions under which the contract incentivizes the user to invest more on protection.
Our analysis is based on comparing the protection investment actions of the user when he purchases and does not purchase the insurance.
The former action is the optimal solution to (\ref{eq:hidden action contract problem 3}), which we denote as $x^*$.
The latter action is the optimal solution to (\ref{eq:action without insurance}), which we denote as $x^0$.

Next, we present our result which characterizes the conditions on the risk preferences for the contract to be incentivizing.

\begin{theorem}
\label{thm:security}
Insurance enhances system security by incentivizing the user to invest more on protection, \ie, $x^*\geq x^0$, if the following conditions hold:\\
(C1) The insurer is more risk-averse than the user: $\rho^i[\xi]\geq \rho^u[\xi]$, for all $x\in X$.\\
(C2) The insurer's risk perception is more sensitive than the user's to the stochastic dominance shift induced by the user's action, \ie, $|\frac{\partial}{\partial x}\rho^i[\xi] |\geq |\frac{\partial}{\partial x}\rho^u[\xi]|$.\\
\end{theorem}
\begin{proof}
We first explicitly express the first-order optimality conditions that $x^0$ and $x^*$ satisfy.
From (\ref{eq:action without insurance}), we observe that $x^0$ satisfies the following equation:
\begin{equation}
    0=\frac{\partial}{\partial x}\rho^u[\xi]+m.
    \label{eq:foc of x0}
\end{equation}
To obtain the formula that $x^*$ has to satisfy, we substitute (\ref{eq:foc of IC}) into (\ref{eq:hidden action contract problem 2}) to eliminate the decision variable $c$:
\begin{equation}
    \min_{x\in X} \  \ U'(x):=(1+\frac{m}{\frac{\partial}{\partial x}\rho^u[\xi]})\rho^i[\xi]-\frac{m}{\frac{\partial}{\partial x}\rho^u[\xi]}\rho^u[\xi]+mx.
    \label{eq:equivalent optimization}
\end{equation}
The optimal solution to (\ref{eq:equivalent optimization}) is equivalent to $x^*$.
By adding the term $\rho^u[\xi]-\rho^u[\xi]$ to $U'(x)$ and performing a few algebraic manipulation, we arrive at the following:
\begin{equation}
    U'(x)=\left(\rho^i[\xi]-\rho^u[\xi] \right)\cdot \left(1+\frac{m}{\frac{\partial}{\partial x}\rho^u[\xi]} \right)+\rho^u[\xi] + mx.
    \label{eq:equivalent user loss}
\end{equation}
Let $D(x)=\rho^i[\xi]-\rho^u[\xi]$.
From (\ref{eq:equivalent user loss}), we obtain the following first-order optimality condition which $x^*$ satisfies:
\begin{equation}
    \begin{aligned}
    0=&\frac{\dd}{\dd x}D(x)\cdot \left(1+\frac{m}{\frac{\partial}{\partial x}\rho^u[\xi]} \right)+D(x)\cdot \frac{-m\frac{\partial^2}{\partial x^2}\rho^u[\xi]}{(\frac{\partial}{\partial x}\rho^u[\xi])^2} \\
    &+\frac{\partial }{\partial x}\rho^u[\xi]+m.
    \label{eq:foc of x*}
    \end{aligned}
\end{equation}
Next, we investigate the relation between (\ref{eq:foc of x0}) and (\ref{eq:foc of x*}).
Fix a subgradient in (\ref{eq:subdifferential of rho}), we can obtain the following:
\begin{equation}
    \frac{\partial^2}{\partial x^2}\rho^u[\xi]=\int_{\Xi}\xi\zeta(\xi)\frac{\partial^2}{\partial x^2}p(\xi,x)\dd P(\xi),
    \label{eq:second order derivative}
\end{equation}
where $\zeta(\xi)\in \mathfrak{A}^*_u$ is a density function with respect to $P(\xi)$.
Since $\Xi\in\mathbb{R}_+$, $\zeta$ is a probability density function and $P$ is a probability measure in
(\ref{eq:second order derivative}), the convexity assumption $\frac{\partial^2}{\partial x^2}p(\xi,x)\geq 0$ leads to $\frac{\partial^2}{\partial x^2}\rho^u[\xi]\geq 0$ for all $x\in X$.
Since (\ref{eq:foc of x0}) contains $\frac{\partial}{\partial x}\rho^u[\xi]$ evaluated at $x^0$ and (\ref{eq:foc of x*}) contains $\frac{\partial}{\partial x}\rho^u[\xi]$ evaluated at $x^*$, $\frac{\partial^2}{\partial x^2}\rho^u[\xi]\geq 0$ tells that to prove $x^*\geq x^0$, it suffices to prove that the following condition holds at $x^*$:
\begin{equation}
    \frac{\dd}{\dd x}D(x)\cdot \left(1+\frac{m}{\frac{\partial}{\partial x}\rho^u[\xi]} \right)+D(x)\cdot \frac{-m\frac{\partial^2}{\partial x^2}\rho^u[\xi]}{(\frac{\partial}{\partial x}\rho^u[\xi])^2}\leq 0.
    \label{eq:condition for x*>x0}
\end{equation}
The second term of the left-hand side of (\ref{eq:condition for x*>x0}) is negative if (C1) holds.
Since $\frac{\partial}{\partial x}\rho^i[\xi]\leq 0 $ and $\frac{\partial}{\partial x}\rho^u[\xi]\leq 0 $ because of Lemma \ref{lemma:consistency}, the first term of the left-hand side of (\ref{eq:condition for x*>x0}) is negative if condition (C2) holds.
This completes the proof.
\end{proof}

Theorem \ref{thm:security} indicates that cyber insurance incentivizes the user's protection investment when conditions (C1) and (C2) hold.
Condition (C1) means that the insurer is more risk-averse than the user.
In other words, the insurer measures the risk to be severer than the user does when they face the same random cyber loss.
This condition suggests that the insurer is more sophisticated in estimating the harm from potential cyber attacks or system failures.  
Hence, the insurer can leverage the more detailed knowledge of how the cyber loss influence the system to design the insurance contract for protection purposes.
Condition (C2) indicates that the insurer's risk perception is more sensitive to the changes of the action of the user than the user himself.
This requirement can have the interpretation that the insurer is more aware of how the probabilistic changes of the cyber losses induced by the user's actions influences the systemic risks.
This privilege of the insurer aides her in designing the incentive contracts so that the user's actions after purchasing the insurance stochastically shift the cyber losses to benefit the overall system security.

\begin{remark}
Risk-aversion of the user is an essential property to allow the user to seek cyber insurance.
For example, the recent work \cite{khalili2018designing} has investigated a linear cyber insurance problem where the insurer is risk neutral.
The authors in \cite{khalili2018designing} have shown that the cyber insurance market only exists when the user is risk-averse.
While condition (C1) in Theorem \ref{thm:security} seems to stand on the opposite side of the result of \cite{khalili2018designing}, we can show that they are consistent.
The reason lies in the difference between minimizing the loss and maximizing the utility. The axioms of risk measures in Definition \ref{def:CRM} only fit minimization problem. 
For maximization problem, condition (A1) becomes condition (A1'): $\rho(Z)\leq \rho(Z')$ if $Z(\xi)\geq Z'(\xi)$ for almost everywhere $\xi\in \Xi$.
When we consider our problem in a maximization setting with risk measures $\rho^{i'}$ and $\rho^{u'}$ satisfying (A1'), the relations between stochastic dominance and the risk measures $\rho^{i'}$ and $\rho^{u'}$ are reversed. 
Then, in the maximization counterpart of our problem, condition (C1) becomes the relation $\rho^{i'}[-\xi]-\rho^{u'}[-\xi]\leq 0$, which means that, observing the same gain $-\xi$, the user should be the one who exhibits more risk-aversion.
Therefore, our setting is consistent with the findings in the literature \cite{khalili2018designing,marotta2017cyber,bohme2010modeling}
\end{remark}

\begin{remark}
When condition (C1) in Theorem \ref{thm:security} holds, 
a decrease in the coverage rate $c$ will lead to a decrease in the objective value of (\ref{eq:hidden action contract problem 2}).
However, this fact does not necessarily leads to a trivial coverage plan, \ie, $c^*=0$, due to the existence of the IC constraint in (\ref{eq:hidden action contract problem 2}).
Suppose that the insurer deceases the coverage plan $c$. 
Since the action of the user $x$ and the coverage plan $c$ should satisfy (\ref{eq:foc of IC}) and the risk $\rho^u$ is decreasing in $x$, the action $x$ of the user will also decrease.
Then, the deceased action $x$ can increase the objective value of (\ref{eq:hidden action contract problem 2}) since both $\rho^i$ and $\rho^u$ increases.
Thus, a non-trivial coverage plan can exist when condition (C1) holds.
\end{remark}

\subsection{Perception Compromise}
\label{sec:analysis:compromise}
In this section, we discuss the interpretations of problem (\ref{eq:hidden action contract problem 2}) by recognizing equivalences in the insurer's choices.

One observes that choosing the coverage rate $c\in[0,1]$ in the objective function of (\ref{eq:hidden action contract problem 2}) can be identified as choosing a distribution with support $\{\rho^i[\xi],\rho^u[\xi]\}$.
This viewpoint can have the interpretation that the insurer balances between her own risk perception $\rho^i$ and the user's risk perception $\rho^u$.  
In other words, in order to design an optimal cyber insurance contract, the insurer compromises her own risk perception to the user's risk perception and tries to find an averaged perception in between them as the calibration of her risk.

The following explanation can also lead to the above observations.
Recall that any convex combination of CRMs is still a CRM \cite{acerbi2002spectral}.
Then, problem (\ref{eq:hidden action contract problem 2}) can be reformulated as:
\begin{equation}
    \begin{aligned}
      \min_{x\in X,\Tilde{\rho}\in \Sigma} \  \ &\Tilde{\rho}[\xi]+mx \\
         \text{s.t.} \  \
         & x\in \argmin_{x'\in X} (1-c)\rho^u[\xi]+mx, \  \ \text{(IC)},
         \label{eq:equivalent risk measure}
    \end{aligned}
\end{equation}
where $\Sigma:=\{\rho:\mathcal{Z}\rightarrow\mathbb{R}|\rho=c\rho^i+(1-c)\rho^u\}$ denotes the set of compromised risk measures.
Problem (\ref{eq:equivalent risk measure}) indicates that the insurer's contract design problem is equivalent to a risk preference design problem.
The set $\Sigma$ means that the feasible choices of risk preferences is a compromise of the insurer's original preference $\rho^i$ to the user's preference $\rho^u$.

\section{Case Study: Ransomware on Computer Systems}
\label{sec:case}
We consider a scenario where a user aims to insure his networked computer system containing $n$ computers against a ransomware.
The ransomware is developed by an attacker, who can hack into the computers to lock down their performances.
Once a computer is locked down by the ransomware, the only way to restart it is to pay a certain amount of ransom to the attacker to let him unlock the computer.
The protection investment $x\in[0,1]$ in this case refers to the level of effort the user spends to update the firewalls of the computer system to defend against the ransomware.
An increase in the protection investment decreases the probability that a computer is infected by the ransomware.
We normalize the ransom to be paid per computer to an amount so that the loss to the user can be denoted by the number of computers that get locked down, \ie, $\Xi=\{1,2,...,n\}$ and the distribution of disabled computers follow a binomial distribution with $n$ trials and success probability $1-0.8x^2$. Note that this setting satisfies the assumption that an increase in the investment induces a stochastic dominance shift in the cyber loss. The factor $0.8$ represents that fact that the risk cannot be fully avoided.

\begin{figure}[ht]
\centering
\vspace{-3mm}
\subfigure[Change of feasible coverage rate when user becomes more risk-averse.]{\includegraphics[width=0.22\textwidth]{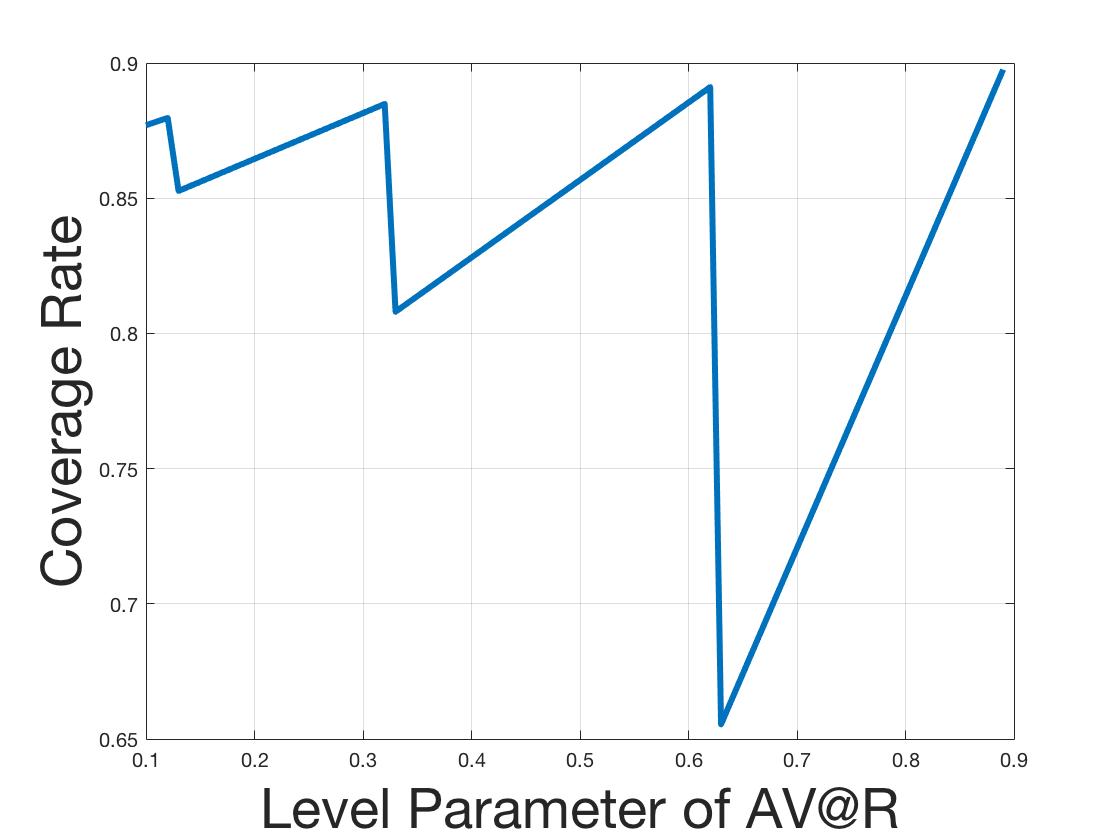}
    \label{fig:subfig1}
}
\ \
\subfigure[Change of feasible premium when user increases investment.]{\includegraphics[width=0.22\textwidth]{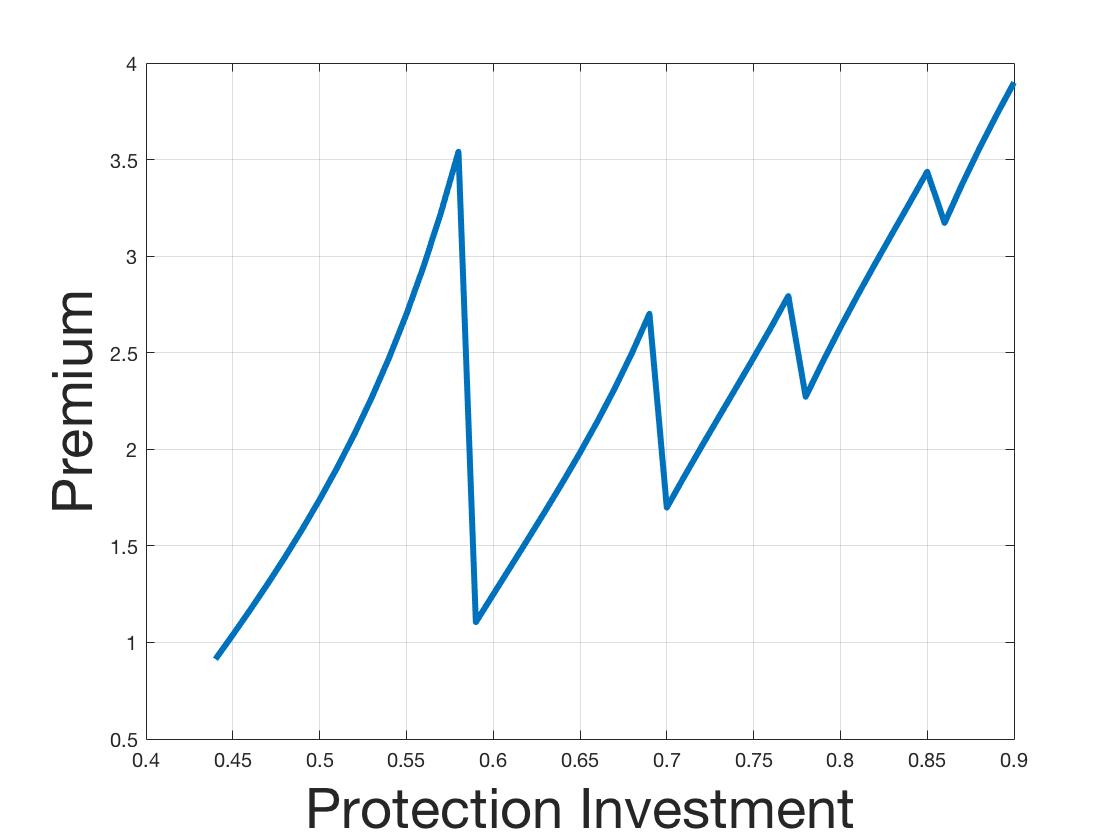}
    \label{fig:subfig2}
}
\vspace{-2mm}
\caption[Optional caption for list of figures]{Properties of feasible contracts.} 
\label{fig:1}
\end{figure}

In numerical results on feasible contracts, we use AV@R to capture the risk perception of the user. 
In Fig. \ref{fig:subfig1}, we discover that the feasible coverage rate is piecewise increasing in the level parameter of the user's AV@R increases.
In Fig. \ref{fig:subfig2}, we show the feasible premium as a function of the user's protection investment. The plot in Fig. \ref{fig:subfig2} also appears to be a piecewise increasing function.
The piecewise increasing property is related to the dual representation of a risk measure as shown in \ref{eq:dual representation of individual risk}.
According to \cite{liu2022mitigating}, since the probability density function which admits the optimal value of \ref{eq:dual representation of individual risk} may lack continuity, the monotonicity property of the contract can only be guaranteed in the region where the density is continuous.
The piecewise increasing property can create challenges for the insurer in design the optimal contract. 
Nevertheless, the contract is increasing locally.
Hence, the insurance can be design more conveniently when the risk perceptions and the actions of the user are in restricted regions.
Such scenarios occur, for example, when the user's firewall updates have limited effects on the security of the computer system, or when the risk perception of the user is resistant to exogenous information.

We can also consider a multi-users scenario where each user owns a subset of the computers and possesses a unique risk perception. 
The risk preference type setting considered in \cite{liu2022mitigating} is a feasible formulation for this scenario.
Accordingly, the insurer designs the contract by considering the averaged response from the population of the users.
We will leave this extension to the future work.

\section{Conclusion}
\label{sec:conclusion}
In this paper, we have formulated a cyber insurance contract design problem which incorporates coherent risk measures to capture the risk perceptions of the insurer and the user.
Using the proposed framework, we have characterized the feasible coverage rate and premium in terms of the risk of the user and its derivative with respect to the user's action.
Furthermore, we have proved that a cyber insurance can incentivize the user to invest more on the system protection when the insurer is both more risk-averse than the user and more sensitive to the impacts of the user's action changes on the risks than the user.

The nonlinearity induced by the risk measures makes the investigation of the optimal contract challenging. 
We leave the study of the properties of the optimal contracts in the future work.
Another possible extension would be the study of insurance design in a dynamic environment. This extension would encounter two major challenges. The first challenge centers around solving a dynamic principal-agent problem. The second challenge arises since we need to adopt a dynamic assessment of the risk. Available approaches to dynamic risk assessment include using the dynamic risk measures or considering a repetition of static risk measures.

\bibliographystyle{IEEEtran}
\bibliography{bibliography.bib}

\begin{thebibliography}{10}
\providecommand{\url}[1]{#1}
\csname url@samestyle\endcsname
\providecommand{\newblock}{\relax}
\providecommand{\bibinfo}[2]{#2}
\providecommand{\BIBentrySTDinterwordspacing}{\spaceskip=0pt\relax}
\providecommand{\BIBentryALTinterwordstretchfactor}{4}
\providecommand{\BIBentryALTinterwordspacing}{\spaceskip=\fontdimen2\font plus
\BIBentryALTinterwordstretchfactor\fontdimen3\font minus
  \fontdimen4\font\relax}
\providecommand{\BIBforeignlanguage}[2]{{%
\expandafter\ifx\csname l@#1\endcsname\relax
\typeout{** WARNING: IEEEtran.bst: No hyphenation pattern has been}%
\typeout{** loaded for the language `#1'. Using the pattern for}%
\typeout{** the default language instead.}%
\else
\language=\csname l@#1\endcsname
\fi
#2}}
\providecommand{\BIBdecl}{\relax}
\BIBdecl

\bibitem{marotta2017cyber}
A.~Marotta, F.~Martinelli, S.~Nanni, A.~Orlando, and A.~Yautsiukhin,
  ``Cyber-insurance survey,'' \emph{Computer Science Review}, vol.~24, pp.
  35--61, 2017.

\bibitem{liu2022mitigating}
S.~Liu and Q.~Zhu, ``Mitigating moral hazard in cyber insurance using risk
  preference design,'' \emph{arXiv preprint arXiv:2203.12001}, 2022.

\bibitem{tversky1992advances}
A.~Tversky and D.~Kahneman, ``Advances in prospect theory: Cumulative
  representation of uncertainty,'' \emph{Journal of Risk and uncertainty},
  vol.~5, no.~4, pp. 297--323, 1992.

\bibitem{stole2001lectures}
L.~Stole, ``Lectures on the theory of contracts and organizations,''
  \emph{Unpublished monograph}, 2001.

\bibitem{rockafellar2000optimization}
R.~T. Rockafellar, S.~Uryasev \emph{et~al.}, ``Optimization of conditional
  value-at-risk,'' \emph{Journal of risk}, vol.~2, pp. 21--42, 2000.

\bibitem{acerbi2002spectral}
C.~Acerbi, ``Spectral measures of risk: A coherent representation of subjective
  risk aversion,'' \emph{Journal of Banking \& Finance}, vol.~26, no.~7, pp.
  1505--1518, 2002.

\bibitem{shapiro2013kusuoka}
A.~Shapiro, ``On kusuoka representation of law invariant risk measures,''
  \emph{Mathematics of Operations Research}, vol.~38, no.~1, pp. 142--152,
  2013.

\bibitem{holmstrom1979moral}
B.~Holmstr{\"o}m, ``Moral hazard and observability,'' \emph{The Bell journal of
  economics}, pp. 74--91, 1979.

\bibitem{khalili2018designing}
M.~M. Khalili, P.~Naghizadeh, and M.~Liu, ``Designing cyber insurance policies:
  The role of pre-screening and security interdependence,'' \emph{IEEE
  Transactions on Information Forensics and Security}, vol.~13, no.~9, pp.
  2226--2239, 2018.

\bibitem{bohme2010modeling}
R.~B{\"o}hme, G.~Schwartz \emph{et~al.}, ``Modeling cyber-insurance: towards a
  unifying framework.'' in \emph{WEIS}, 2010.

\bibitem{artzner1999coherent}
P.~Artzner, F.~Delbaen, J.-M. Eber, and D.~Heath, ``Coherent measures of
  risk,'' \emph{Mathematical finance}, vol.~9, no.~3, pp. 203--228, 1999.

\bibitem{uryasev2010var}
S.~Uryasev, S.~Sarykalin, G.~Serraino, and K.~Kalinchenko, ``Var vs cvar in
  risk management and optimization,'' in \emph{CARISMA conference}, 2010.

\bibitem{shapiro2021lectures}
A.~Shapiro, D.~Dentcheva, and A.~Ruszczynski, \emph{Lectures on stochastic
  programming: modeling and theory}.\hskip 1em plus 0.5em minus 0.4em\relax
  SIAM, 2021.

\bibitem{ruszczynski2006optimization}
A.~Ruszczy{\'n}ski and A.~Shapiro, ``Optimization of convex risk functions,''
  \emph{Mathematics of operations research}, vol.~31, no.~3, pp. 433--452,
  2006.

\end{thebibliography}
\nocite{*}

%






\end{document}